\documentclass[11pt]{llncs}

\usepackage{amssymb}
\usepackage{amsmath}

\usepackage{tikz}
\usepackage[all]{xy}

\newcommand{\betredge}[1]{\betstep{#1}}
\newcommand{\betstep}[1]{\mbox{$\stackrel{#1}{\rightarrow}$}}

\newcommand{\tlb}{\ensuremath{[\![}}
\newcommand{\trb}{\ensuremath{]\!]}}

\newcommand{\sut}{ \:|\: }
\newcommand{\su}{\subseteq}
\newcommand{\lf}{\footnotesize}

\newcommand{\m}{\mathcal} 
\newcommand{\dg}{\text{deg}}
\newcommand{\cP}{\text{P}}
\newcommand{\SW}{\textit{SW}}

\newcommand{\cedit}
{\textcolor{black}}

\newcommand{\gsrem}[1]
{}

\newcommand{\mrrem}[1]{\marginpar{\tiny\textcolor{blue}{MR: #1}}}

\usepackage{txfontsb}

\title{Coordination Games on Graphs\thanks{An extended abstract of this paper appeared in \cite{ARSS14}. Part of this research has been carried out while the second author was a post-doctoral researcher at Sapienza University of Rome, Italy.}}

\author{%
Krzysztof R. Apt \inst{1}
\and
Bart de Keijzer \inst{1}
\and
Mona Rahn \inst{1}
\and
Guido Sch\"{a}fer \inst{1,2}
\and
Sunil Simon \inst{3}
}

\institute{%
CWI, Amsterdam, The Netherlands
\and
VU University Amsterdam, The Netherlands
\and
IIT Kanpur, Kanpur, India
}

\begin{document}
\sloppy
\maketitle

\begin{abstract}
We introduce natural strategic games on graphs, which capture the idea of coordination in a local setting. We study the existence of equilibria that are resilient to coalitional deviations of unbounded and bounded size (i.e., \emph{strong equilibria} and \emph{$k$-equilibria} respectively). We show that pure Nash equilibria and $2$-equilibria exist, and give an example in which no $3$-equilibrium exists. Moreover, we prove that strong equilibria exist for various special cases.

We also study the price of anarchy (PoA) and price of stability (PoS) for these solution concepts. We show that the PoS for strong equilibria is $1$ in almost all of the special cases for which we have proven strong equilibria to exist. The PoA for pure Nash equilbria turns out to be unbounded, even when we fix the graph on which the coordination game is to be played. For the PoA for $k$-equilibria, we show that the price of anarchy is between $2(n-1)/(k-1) - 1$ and $2(n-1)/(k-1)$. The latter upper bound is tight for $k=n$ (i.e., strong equilibria).

Finally, we consider the problems of computing strong equilibria and of determining whether a joint strategy is a $k$-equilibrium or strong equilibrium. We prove that, given a coordination game, a joint strategy $s$, and a number $k$ as input, it is co-NP complete to determine whether $s$ is a $k$-equilibrium. On the positive side, we give polynomial time algorithms to compute strong equilibria for various special cases.
\end{abstract}

\section{Introduction}


In game theory, coordination games are used to model situations in which players are rewarded for agreeing on a common strategy, e.g., by deciding on a common technological or societal standard. In this paper we introduce and study a very simple  class of coordination games, which we call \emph{coordination games on graphs}: 
\begin{quote}
We are given a finite (undirected) graph, of which the nodes correspond to the players of the game. Each player chooses a color from a set of colors available to her. The payoff of a player is the number of neighbors who choose the same color. 
\end{quote}
Our main motivation for studying these games is that they
constitute a natural class of strategic games that capture the following three key characteristics:
\begin{enumerate}
\item \emph{Join the crowd property} \cite{SA15}: the payoff of each player weakly increases when more players choose her strategy.
\item \emph{Asymmetric strategy sets}: players may have different strategy sets.
\item \emph{Local dependency}: the payoff of each player depends only on the choices made by certain groups of players (i.e., neighbors in the given graph).
\end{enumerate}
The above characteristics are inherent to many applications. As a concrete example,  consider a situation in which several clients have to choose between multiple competing providers offering the same service (or product), such as peer-to-peer networks, social networks, photo sharing platforms, and mobile phone providers. 
Here the benefit of a client for subscribing to a specific provider increases with the number of clients who opt for this provider. Also, each client typically cares only about the subscriptions of certain other clients (e.g., friends, relatives, etc.). 

In coordination games on graphs it is beneficial for each player to align her choices with the ones of her neighbors. As a consequence, the players may attempt to increase their payoffs by coordinating their choices in groups (also called \emph{coalitions}). In our studies we therefore focus on equilibrium concepts that are resilient to deviations of groups; more specifically we study \emph{strong equilibria} \cite{Aumann59} and \emph{$k$-equilibria} (also known as \emph{$k$-strong equilibria}) of coordination games on graphs. 
Recall that in a strong equilibrium no coalition of players can profitably deviate in the sense that every player of the coalition strictly improves her payoff. Similarly, in a $k$-equilibrium with $k \in \{1, \dots, n\}$, where $n$ is the number of players, no coalition of players of size at most $k$ can profitably deviate.

\paragraph{Our contributions.}

The focus of this paper is on the existence, inefficiency and computability of strong equilibria and $k$-equilibria of coordination games on graphs. Our main contributions are as follows: 

\bigskip\noindent
\emph{1. Existence.} 
We show that Nash equilibria and 2-equilibria always exist. On the other hand, $k$-equilibria for $k \geq 3$ do not need to exist.
We therefore derive a complete characterization of the values of $k$ for which $k$-equilibria exist in our games.

We also show that strong equilibria exist if only two colors are available.
Further, we identify several graph structural properties that guarantee the existence of strong equilibria: in particular they exist if the underlying graph is a pseudoforest\footnote{Recall that in a pseudoforest each connected component has at most one cycle.}, and when every pair of cycles in the graph is edge-disjoint. Also, they exist if the graph is \emph{color complete}, i.e., if for each available color $x$ the components of the subgraph induced by the nodes having color $x$ are complete. Moreover, existence of strong equilibria is guaranteed in case the coordination game is played on a \emph{color forest}, i.e., for every color, the subgraph induced by the players who can choose that color is a forest.

We also address the following question. Given a coordination game denote its \emph{transition value} as the value of $k$ for which a $k$-equilibrium exists but a $(k+1)$-equilibrium does not. The question then is to determine for which values of $k$ a game with transition value $k$ exists. We exhibit a game with transition value 4.


In all our proofs the existence of strong equilibria is established by showing a stronger result, namely that the game has the \emph{coalitional finite improvement property}, i.e., every sequence of profitable joint deviations is finite (see Section~\ref{sec:prelim} for a formal definition).

\bigskip\noindent
\emph{2. Inefficiency.}
We also study the \emph{inefficiency} of equilibria. In our context, the \emph{social welfare} of a joint strategy is defined as the sum of the payoffs of all players. The \emph{$k$-price of anarchy} \cite{AFM09} (resp. \emph{$k$-price of stability}) refers to the ratio between the social welfare of an optimal outcome and the minimum (resp.~maximum) social welfare of a $k$-equilibrium\footnote{The $k$-price of anarchy is also commonly known as the \emph{$k$-strong price of anarchy}.}.

We show that the price of anarchy is unbounded, independently of the underlying graph structure, and the strong price of anarchy is $2$. In general, for the $k$-price of anarchy with $k \in \{2, \dots, n-1\}$ we derive almost matching lower and upper bounds of $2\frac{n-1}{k-1} - 1$ and $2\frac{n-1}{k-1}$, respectively (given a coordination game that has a $k$-equilibrium). 
We also prove that the strong price of stability is 1 for the cases that there are only two colors, or the graph is a pseudoforest or color forest.

Our results thus show that as the coalition size $k$ increases, the worst-case inefficiency of $k$-equilibria decreases from $\infty$ to $2$. In particular, we obtain a constant $k$-price of anarchy for $k = \Omega(n)$. 

\bigskip\noindent
\emph{3. Complexity.} 
We also address several computational complexity issues.
Given a coordination game, a joint strategy $s$, and a number $k$ as input, it is co-NP complete to determine whether $s$ is a $k$-equilibrium. 
However, we show that this problem can be solved in polynomial time in case the graph is a color forest.
We also give polynomial time algorithms to compute strong equilibria for the cases of color forests, color complete graphs, and pseudoforests.


\paragraph{Related work.} 

Our coordination games on graphs are related to various well-studied types of games. We outline some connections below. 

First, coordination games on graphs are \emph{polymatrix games}.
Recall that a polymatrix game (see \cite{How72,Jan68}) is a finite
strategic game in which the payoff for each player is the sum of the
payoffs obtained from the individual games the player plays with each
other player separately.
Cai and Daskalakis \cite{CD11} considered a special class of
polymatrix games which they call \emph{coordination-only polymatrix
  games}. 
These games are identical to coordination games on graphs with edge weights. 
They showed that pure Nash equilibria exist and that
finding one is PLS-complete. The proof of the latter result crucially
exploits that the edge weights can be negative. Note that negative
edge weights can be used to enforce that players anti-coordinate. Our
coordination games do not exhibit this characteristic and are
therefore different from theirs.


Second, our coordination games are related to \emph{additively
  separable hedonic games (ASHG)} \cite{Baner,Bogo}, which were
originally proposed in a cooperative game theory setting. Here the
players are the nodes of an edge weighted graph and form coalitions.
The payoff of a node is defined as the total weight of all edges to
neighbors that are in the same coalition. If the edge weights are
symmetric, the corresponding ASHG is said to be \emph{symmetric}.
Recently, a lot of work focused on computational issues of these games
(see, e.g., \cite{ABS10,ABS11,Gair}). Aziz and Brandt \cite{Aziz}
studied the existence of strong equilibria in these games.  The
PLS-hardness result established in \cite{Gair} does not carry over to
our coordination games because it makes use of negative edge weights,
which we do not allow in our model. Note also that in ASHGs every
player can choose to enter every coalition which is not necessarily
the case in our coordination games. Such restrictions can be imposed
by the use of negative edge weights (see also \cite{Gair}) and our
coordination games therefore constitute a special case of symmetric
ASHGs with \emph{arbitrary} edge weights.

Third, our coordination games on graphs are related to
\emph{congestion games} \cite{Ros73}. In particular, they are
isomorphic to a special case of congestion games with weakly
decreasing cost functions (assuming that each player wants to minimize
her cost).  Rozenfeld and Tennenholtz \cite{RT06} derived a structural
characterization of strategy sets that ensure the existence of strong
equilibria in such games. By applying their characterization to our
(transformed) games one obtains that strong equilibria exist if the
underlying graph of the coordination game is a matching or complete
(both results also follow trivially from our studies).  Bil\`o et
al.~\cite{BFFM11} studied congestion games where the players are
embedded in a (possibly directed) \emph{influence graph} (describing
how the players delay each other). They analyzed the existence and
inefficiency of pure Nash equilibria in these games. However, because
the delay functions are assumed to be linearly increasing in the
number of players, these games do not cover the games we study here.

Further, coordination games on graphs are special cases of the
\emph{social network games} introduced and analyzed in \cite{AS13} (if
one uses in them thresholds equal to 0). These are games associated
with a threshold model of a social network introduced in \cite{AM11}
which is based on weighted graphs with thresholds.


Coordination games are also related to the problem of clustering, where
the task is to partition the nodes of a graph in a meaningful manner.
If we view the strategies as possible cluster names, then a Nash
equilibrium of our coordination game on a graph corresponds to a
``satisfactory'' clustering of the underlying graph. Hoefer
\cite{Hoefer2007} studied clustering games that are also polymatrix
games based on graphs. Each player plays one of two possible base
games depending on whether the opponent is a neighbor in the given
graph or not.  Another more recent approach to clustering through
game theory is by Feldman, Lewin-Eytan and Naor \cite{FLN15}. In this
paper both a fixed clustering of points lying in a metric space and a
correlation clustering (in which the distance is in [0,1] and
each point has a weight denoting its `influence') is viewed as a
strategic hedonic game. However, in both references each player has
the same set of strategies, so the resulting games are not comparable
with ours.

Strategic games that involve coloring of the vertices of a graph have
also been studied in the context of the vertex coloring problem.
These games are motivated by the question of finding the chromatic
number of a graph. As in our games, the players are nodes in a graph
that choose colors. However, the payoff function differs from the one
we consider here: it is $0$ if a neighbor chooses the same color and
it is the number of nodes that chose the same color otherwise.
Panagopoulou and Spirakis \cite{PS08} showed that an efficient local
search algorithm can be used to compute a good vertex coloring.
Escoffier, Gourv\`es and Monnot \cite{EGM12} extended this work by
analyzing socially optimal outcomes and strong equilibria.
Chatzigiannakis et al.~\cite{CKPS10} studied the vertex coloring
problem in a distributed setting and showed that under certain
restrictions a good coloring can be reached in polynomial time.

Strong and $k$-equilibria in strategic games on graphs were also
studied in Gourv\`es and Monnot \cite{GM09,GM10}. These games are
related to, respectively, the \texttt{MAX-CUT} and \texttt{MAX-$k$-CUT} problems. However, they do not satisfy the join the crowd
property, so, again, the results are not comparable with ours.

To summarize, in spite of these close connections, our coordination
games on graphs are different from all classes of games mentioned
above. Notably, this is due to the fact that our games combine the
three properties mentioned above, i.e., join the crowd, asymmetric
strategy sets and local dependencies modeled by means of an undirected
graph.

Research reported here was recently followed in two different
directions.  In \cite{ASW15} and \cite{SW16} coordination games
on directed graphs were considered, while in \cite{rahn:schaefer:2015}
coordination games on weighted undirected graphs were analyzed.
Both setups lead to substantially different results that are
  discussed in the final section.
Finally, \cite{FF15} studied the strong price of anarchy for a general
class of strategic games that, in particular, include as special cases
our games and the \texttt{MAX-CUT} games mentioned above.

As a final remark, let us mention that the coordination games on
graphs are examples of games on networks, a vast research area
surveyed in \cite{JZ12}.


\paragraph{Our techniques.} 

Most of our existence results are derived through the application of one technical key lemma. This lemma relates the change in social welfare caused by a profitable deviation of a coalition to the size of a minimum feedback edge set of the subgraph induced by the coalition\footnote{Recall that a \emph{feedback edge set} is a set of edges whose removal makes the graph acyclic.}. This lemma holds for arbitrary graphs and provides a tight bound on the maximum decrease in social welfare caused by profitable deviations. Using it, we prove our existence results by means of a generalized ordinal potential function argument. In particular, this enables us to show that every sequence of profitable joint deviations is finite. Further, we use the generalized ordinal potential function to prove that the strong price of anarchy is 1  and that strong equilibria can be computed efficiently for certain graph classes. 


The non-existence proof of $3$-equilibria is based on an instance whose graph essentially corresponds to the skeleton of an octahedron and whose strategy sets are set up in such a way that at most one facet of the octahedron can be unicolored. We then use the symmetry of this instance to prove our non-existence result.

The upper bound on the $k$-price of anarchy is derived through a combinatorial argument. We first fix an arbitrary coalition of size $k$ and relate the social welfare of a $k$-equilibrium to the social welfare of an optimum within this coalition. We then extrapolate this bound by summing over all coalitions of size at most $k$. We believe that this approach might also prove useful to analyze the $k$-price of anarchy in other contexts. 


\section{Preliminaries}
\label{sec:prelim}

A \emph{strategic game} $\mathcal{G} := (N, (S_i)_{i \in N}, (p_i)_{i \in N})$ consists of a set $N := \{1, \dots, n\}$ of $n > 1$ players, a non-empty set $S_i$ of \emph{strategies}, and a \emph{payoff function} $p_i : S_1 \times \cdots \times S_n \rightarrow \mathbb{R}$ for each player $i \in N$.
We denote $S_1 \times \cdots \times S_n$ by $S$, 
call each element $s \in S$
a \emph{joint strategy}, and abbreviate the sequence
$(s_{j})_{j \neq i}$ to $s_{-i}$. Occasionally we write $(s_i,
s_{-i})$ instead of $s$.  

We call a non-empty subset $K := \{k_1, \ldots, k_m\}$ of $N$ a \emph{coalition}. Given a joint strategy $s$ we abbreviate the sequence $(s_{k_1}, \ldots, s_{k_m})$ of strategies to $s_K$ and $S_{k_1} \times \cdots \times S_{k_m}$ to $S_{K}$. We also write $(s_K, s_{-K})$ instead of $s$. If there is a strategy $x$ such that $s_i = x$ for all players $i \in K$, we also write $(x_K, s_{-K})$ for $s.$




Given two joint strategies $s'$ and $s$ and a coalition $K$, we say that $s'$ is a \emph{deviation of the players in $K$}
from $s$ if $K = \{i \in N \mid s_i \neq s_i'\}$.  
We denote this by $s \betredge{K} s'$. If in addition $p_i(s') > p_i(s)$ holds for all $i \in K$, we say that the deviation $s'$ from $s$ is \emph{profitable}. Further, we say that the players in $K$ \emph{can profitably deviate from $s$} if there exists a profitable deviation of these players
from $s$. 

Next, we call a joint strategy $s$ a \emph{k-equilibrium}, where $k \in \{1, \dots, n\}$, if no coalition of at most $k$ players can profitably deviate from $s$. 
Using this definition, a \emph{Nash equilibrium} is a 1-equilibrium and a \emph{strong equilibrium} \cite{Aumann59} is an $n$-equilibrium. 

Given a joint strategy $s$, we call the sum $\mathit{\SW}(s)=\sum_{i \in N} p_i(s)$ the \emph{social welfare} of $s$. When the social welfare of $s$ is maximal, we call $s$ a \emph{social optimum}. 
Given a finite game that has a $k$-equilibrium, its \emph{$k$-price of anarchy (resp. stability)} is the ratio $\mathit{\SW}(s)/\mathit{\SW}(s')$, where $s$ is a social optimum and $s'$ is a $k$-equilibrium with the lowest (resp. highest) social welfare\footnote{In the case of division by zero, we define the outcome as $\infty$.}. The \emph{(strong) price of anarchy} refers to the $k$-price of anarchy with $k = 1$ ($k = n$). The \emph{(strong) price of stability} is defined analogously. 

A \emph{coalitional improvement path}, in short a \emph{c-improvement
  path}, is a maximal sequence $(s^1, s^2, \dots)$ of joint strategies
such that for every $k > 1$ there is a coalition $K$ such that $s^k$
is a profitable deviation of the players in $K$ from
$s^{k-1}$. Clearly, if a c-improvement path is finite, its last
element is a strong equilibrium. We say that $\mathcal{G}$ has the
\emph{finite c-improvement property} (\emph{c-FIP}) if every
c-improvement path is finite. So if $\mathcal{G}$ has the c-FIP, then
it has a strong equilibrium. Further, we say that the function $P: S
\rightarrow A$ (where $A$ is any set) is a \emph{generalized ordinal
  c-potential} for $\m G$ if there exists a strict partial
ordering $\succ$ on the set $A$ such that if $s \betredge{K} s'$ is a profitable deviation, then $P(s') \succ
P(s)$. A generalized ordinal potential is also called a \emph{generalized strong potential}
\cite{Harks13,Holzman97}.
%
It is easy to see that if a finite game admits a generalized ordinal c-potential then the game has the c-FIP. 
The converse also holds: a finite game that has the c-FIP admits a generalized ordinal c-potential. The latter fact is folklore; we give a self-contained proof in Appendix~\ref{app:gen-ord-potential}. 

Note that in the definition of a profitable deviation of a coalition,
we insisted that all members of the coalition change their strategies.
This requirement is irrelevant for the definitions of the $k$-equilibrium and the c-FIP, but it makes some
arguments slightly simpler.

\section{Coordination games on graphs}
\label{sec:coloring}
We now introduce the games we are interested in. Throughout the paper, we fix a finite set of colors $M$ of size $m$, an undirected graph $G = (V,E)$ without self-loops, and a \emph{color assignment} $A$. The latter is a function that assigns to each node $i$ a non-empty set $A_i \subseteq M$.
A node $j \in V$ is a \emph{neighbor} of the node $i \in V$ if $\{i,j\} \in E$.
Let $N_i$ denote the set of all neighbors of node $i$. 
We define a strategic game $\mathcal{G}(G, A)$ as follows:
\begin{itemize}
\item the players are identified with the nodes, i.e., $N = V$,
  
\item the set of strategies of player $i$ is 
  $A_i$,

\item the payoff function of player $i$ is $p_i(s) := |\{ j \in N_i \mid s_i = s_j \} |$.

\end{itemize}

\noindent So each node simultaneously chooses a color from the set available to her and the payoff to the node is the number of neighbors who chose the same color. 
We call these games \emph{coordination games on graphs}, from now on just \emph{coordination games}.

\begin{figure}[t]
\begin{center}
  \begin{tikzpicture}

    \node [label={\lf $\{ \underline{a}, c\}$}] (1) at (-4, 0) {$1$};       
    \node [label={\lf $\{a, \underline{b}\}$}] (2) at (-2, 1) {$2$}; 
    \node [label=below:{\lf $\{a, \underline{b}\}$}] (3) at (-2, -1) {$3$};     
    \node [label={\lf $\{\underline{b}, c\}$}] (4) at (0, 1) {$4$};
    \node [label=below:{\lf $\{\underline{b}, c\}$}] (5) at (0, -1) {$5$};
    \node [label={\lf $\{c, \underline{a}\}$}] (6) at (2, 1) {$6$};
    \node [label=below:{\lf $\{c, \underline{a}\}$}] (7) at (2,-1) {$7$};
    \node [label={\lf $\{b, \underline{a}\}$}] (8) at (4, 0) {$8$};   
       \draw (1) -- (2);       
       \draw (1) -- (3);              
       \draw[ultra thick] (2) -- (4);                            
       \draw[ultra thick] (2) -- (5);
       \draw[ultra thick] (3) -- (4);                            
       \draw[ultra thick] (3) -- (5);       
       \draw[ultra thick] (4) -- (5);       
       \draw (4) -- (1);              
       \draw (4) -- (8);                     
       \draw (5) -- (1);              
       \draw (5) -- (8);                     
       \draw (6) -- (4);                            
       \draw (6) -- (5);                     
       \draw (7) -- (4);                            
       \draw (7) -- (5);       
       \draw[ultra thick] (8) -- (6);       
       \draw[ultra thick] (8) -- (7);              
       \path[ultra thick] (1) edge [out=-75, in=-105] node {} (8);
\end{tikzpicture}
\caption{A graph with a color set assignment. The bold edges indicate pairs of players choosing the same color.}
\label{fig:non-existence}
\end{center}
\end{figure}
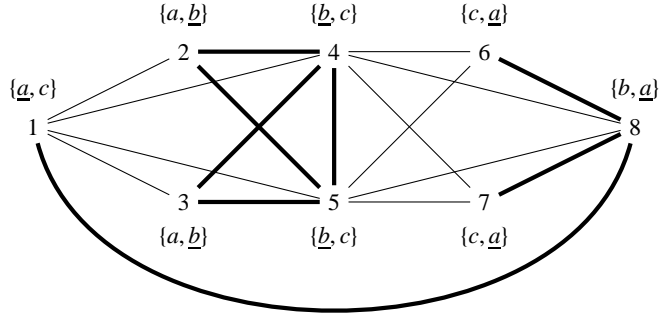

\begin{example}\label{exa:1}
Consider the graph and the color assignment depicted in
Figure~\ref{fig:non-existence}. 
Take the joint strategy that consists of the underlined strategies. Then the payoffs are as follows:

\begin{itemize}
\item 1 for the nodes 1, 6, 7,

\item 2 for the nodes 2, 3,

\item 3 for the nodes 4, 5, 8.
\end{itemize}

It is easy to see that the above joint strategy is a Nash equilibrium.
However, it is not a strong equilibrium because the coalition $K =
\{1, 4, 5, 6, 7\}$ can profitably deviate by choosing color $c$.
\qed
\end{example}

We now recall some notation and introduce some terminology.
Let $G = (V, E)$ be a graph. Given a set of nodes $K$, we denote by $G[K]$
the subgraph of $G$ induced by $K$ and by $E[K]$ the set of edges in $E$ that
have both endpoints in $K$. So $G[K] = (K, E[K])$.  Further,
$\delta(K)$ denotes the set of edges that have one node in $K$ and
the other node outside of $K$. Also, given a subgraph $C$ of $G$ we use $V(C)$ and $E(C)$ to refer to the set of nodes and the set of edges of $C$, respectively. 

Furthermore, we define $\SW_K(s) := \sum_{i \in K} p_i(s)$. Given a joint strategy $s$ we denote by
$E^+_s$ the set of edges $\{i,j\} \in E$ such that $s_i = s_j$. We call
these edges \emph{unicolored in $s$.} (In Figure \ref{fig:non-existence}, these are the bold edges.) Note that $\SW(s) = 2|E_s^+|$.
Finally, we call a subgraph \emph{unicolored in $s$} if all its nodes
have the same color in $s$.

\section{Existence of strong equilibria}
\label{sec:existence}

We begin by studying the existence of strong equilibria and $k$-equilibria of coordination games. We first prove our key lemma and then show how it can be applied to derive several existence results. 

\subsection{Key lemma}

Recall that an edge set $F \su E$ is a \emph{feedback edge set} of the graph $G = (V, E)$ if the graph $(V, E \setminus F)$ is acyclic.

\begin{lemma}[Key lemma]
\label{lem:tau}
Suppose $s \betredge{K} s'$ is a profitable deviation. Let $F$ be a feedback edge set of $G[K]$. Denote $\SW(s') - \SW(s)$ by $\Delta \SW$ and for a coalition $L$ denote $\SW_L(s) - \SW_L(s)$ by $\Delta \SW_L$. Then 
\begin{equation}\label{eqn:deltasw}
\Delta \SW = 2 (\Delta \SW_K - |E^+_{s'} \cap E[K]| + |E^+_s \cap E[K]|)
\end{equation}
and 
\begin{equation}\label{eqn:deltasw2} 
\Delta \SW > 2 (|F \cap E_s^+| - |F \cap E_{s'}^+|).
\end{equation}

\end{lemma}
\begin{proof}
  Let $N_K$ denote the set of neighbors of nodes in $K$ that are not
  in $K$.  Abbreviate $\SW(s')-\SW(s)$ to $\Delta \SW$, and
  analogously, for a coalition $L$, let $\Delta
  \SW_{L}=\SW_{L}(s')-\SW_{L}(s)$.  The change in the social welfare
  can be written as
$$
\Delta \SW = \Delta \SW_K + \Delta \SW_{N_K} + \Delta \SW_{V \setminus (K \cup N_K)}.
$$
We have
$$
\SW_K(s) = 2 |E^+_s \cap E[K]| + |E^+_s \cap \delta(K)|
$$
and analogously for $s'$. Thus
\begin{align*}
\Delta \SW_K =  2(|E^+_{s'} \cap E[K]| - |E^+_s \cap E[K]|) + |E^+_{s'} \cap \delta(K)| - |E^+_s \cap \delta(K)|.
\end{align*}
It follows that 
\begin{align*}
\Delta \SW_{N_K} &= |E^+_{s'} \cap \delta(K)| - |E^+_s \cap \delta(K)| \\ 
&= \Delta \SW_K - 2(|E^+_{s'} \cap E[K]| - |E^+_s \cap E[K]|).
\end{align*}
Furthermore, the payoff of the players that are neither in
$K$ nor in $N_K$ does not change and hence $\Delta \SW_{V \setminus (K \cup N_K)} = 0.$
Putting these equalities together, we obtain \eqref{eqn:deltasw}.

Let $F^c = E[K] \setminus F$. Then
\begin{align*}
  |E^+_s \cap E[K]| - |E^+_{s'} \cap E[K]|&= |E^+_s \cap F|  -|E^+_{s'} \cap F| + |E^+_s \cap F^c| - |E^+_{s'} \cap F^c|.
\end{align*}
We know that $(K, F^c)$ is a forest because $F$ is a feedback edge set. So $|F^c| < |K|.$ Hence
\begin{align*}
|E^+_s \cap F^c| - |E^+_{s'} \cap F^c| \geq - |F^c| >  - |K|.
\end{align*}
Furthermore, each player in $K$ improves his payoff when switching to
$s'$ and hence $\Delta \SW_K \geq |K|.$ So, plugging in these
inequalities in \eqref{eqn:deltasw} we get
$$
\Delta \SW > 2(|K| + |E^+_s \cap F| - |E^+_{s'} \cap F|  - |K|) = 2(|E^+_s \cap F| - |E^+_{s'} \cap F|),
$$
which proves \eqref{eqn:deltasw2}.
 \qed
\end{proof}

Let $\tau(K)$ be the size of a minimal feedback edge set of $G[K]$,
i.e.,
\begin{equation}\label{def:minFES}
\tau(K) = \min\{|F| \sut G[K] \setminus F \text{ is acyclic}\}.
\end{equation}
Equation \eqref{eqn:deltasw2} then yields that $\SW(s') - \SW(s) > - 2 \tau(K)$. 
The following example shows that this bound is tight.

\begin{example} \label{exa:mona} 
We define a graph $G = (V, E)$ and a color assignment as follows. 
Consider a clique on $l$ nodes and let $K$ be the set of nodes. 
Every $i \in K$ can choose between two colors $\{c_i, x\}$, where $c_i \neq c_j$ for every $j \neq i$. Further, every node $i \in K$ is adjacent to $(l-2)$ additional nodes of degree one, each of which has the color set $\{c_i\}$. Note that when defining a joint strategy $s$, it is sufficient to specify $s_i$ for every $i \in K$ because the remaining nodes have only one color to choose from. 

Let $s:=(c_i)_{i \in K}$ and $s' := (x)_{i \in K}$. Then $s \betredge{K} s'$ is a
profitable deviation because every node in $K$ increases its payoff
from $(l-2)$ to $(l-1).$ Also $ |E^+_{s}| = l(l-2)$ and $ |E^+_{s'}| =
\frac{l(l-1)}{2}$, so
\[
|E^+_s| - |E^+_{s'}| = l \left(l-2-\frac{l-1}{2} \right) = l \left( \frac{l-1}{2}-1 \right).
\]
Furthermore, each tree on $|K|$ nodes has $|K|-1$ edges. Thus
\[
\tau(K) = |E[K]| - (|K|-1) =  \frac{l(l-1)}{2}  - (l-1) = l \left( \frac{l-1}{2} -1 \right) + 1.
\]
So $\SW(s') - \SW(s) = 2(|E^+_{s'}| - |E^+_s|) = -2\tau(K) + 2$. Tightness follows because the left hand side is always even. 
\qed
\end{example}

\subsection{Color forests and pseudoforests} 

We use our key lemma to show that coordination games on pseudoforests admit strong equilibria. 
Recall that a \emph{pseudoforest} is a graph in which every connected component contains at most one cycle. For a color $x \in M$ let 
\[
V_x = \{i \in V \mid x \in A_i\}
\] 
be the set of nodes that can choose $x.$  If $G[V_x]$ is a forest for all $x \in M$, we call $G$ a \emph{color forest} (with respect to $A$). Note that, in particular, a forest constitutes a color forest.
Given a joint strategy $s$, we call a subgraph $G'$ of $G$ \emph{completely non-unicolored in $s$} if none of its edges is unicolored in $s$. 

We first derive some corollaries from our key lemma. Throughout this section, we consider a profitable deviation $s \betredge{K} s'$ and let $\Delta \SW = \SW(s') - \SW(s)$. 

\begin{corollary} \label{cor:no_unicolor_cycle}  
If $\Delta \SW \leq 0$, then there is a cycle $C$ in $G[K]$ that is completely non-unicolored in $s$ and unicolored in $s'$. 
\end{corollary}
\begin{proof} 
Assume that the claim does not hold. Then for all cycles $C$ in $G[K]$, we can pick an edge $e_C \in E[C]$ that is unicolored in $s$ or non-unicolored in $s'$. Let
$
F = \{e_C \sut C \text{ is a cycle in } G[K]\}.
$
This is a feedback edge set satisfying $F \cap E_{s'}^+ \subseteq F \cap E_{s}^+$. Hence by \eqref{eqn:deltasw2},
$
\Delta \SW > 2(|E^+_s \cap F| - |E^+_{s'} \cap F|) \geq 0,
$
which is a contradiction.  
\qed
\end{proof}
The next statement follows immediately from Corollary \ref{cor:no_unicolor_cycle} because unicolored cycles cannot exist in color forests. Note that forests are a special case. 

\begin{theorem} \label{thm:color_forest} 
Suppose that $G[K]$ is a color forest. Then  $\Delta \SW  > 0.$ Hence every coordination game on a color forest has the c-FIP. \qed
\end{theorem} 

\begin{corollary} \label{cor:one-cycle}
If $G[K]$ is a graph with at most one cycle, then $\Delta \SW  \geq 0.$  
\end{corollary}
\begin{proof}
  If $G[K]$ is a connected graph with exactly one cycle, then there is a feedback edge set of size $1$. Hence $\Delta \SW > -2$. Because the left hand side
  is even, this implies $\Delta \SW \geq 0.$
\qed
\end{proof}

Using Corollary \ref{cor:no_unicolor_cycle} and \ref{cor:one-cycle},
we now establish the following result.

\begin{theorem} \label{thm:pseudoforest}
  Every coordination game on a pseudoforest has the c-FIP. 
\end{theorem}
\begin{proof}
Associate with each joint strategy $s$ the pair
\[
P(s): = \left(\SW(s),\ |\{ C \mid \text{$C$ is a unicolored cycle in $s$} \} | \right).
\]
We now claim that
  $P: S \to \mathbb{R}^2$ is a generalized ordinal c-potential
  when we take for the strict partial ordering $\succ$ on $P(S)$
the lexicographic ordering.

Consider a profitable deviation $s \betredge{K}
s'$. By partitioning $K$ into the subsets of different
connected components we can decompose this deviation into a sequence of profitable deviations such that each
deviating coalition induces a subgraph of a connected graph with at most one cycle.
By Corollary \ref{cor:one-cycle} the social
welfare in each of these profitable deviations weakly increases.  So
$\SW(s') \geq \SW(s)$. 

If $\SW(s') > \SW(s)$ then $P(s') \succ P(s)$.  If $\SW(s') = \SW(s)$, then by Corollary \ref{thm:color_forest} each of these profitable deviations is by a coalition that induces a connected graph with exactly one cycle. Moreover, this cycle becomes unicolored in $s'$. Thus $P(s') \succ P(s)$.
\qed
\end{proof}

\subsection{Further applications}
The following corollary is an immediate consequence of Corollary~\ref{thm:color_forest}. 

\begin{corollary} \label{cor:2-eq}
In every coordination game, every sequence of profitable deviations of coalitions of size at most two is finite. Hence Nash equilibria and 2-equilibria always exist.
\end{corollary}

\begin{corollary} 
Every coordination game in which at most two colors are used has the c-FIP.
\end{corollary}
\begin{proof}
  Let $s \betredge{K} s'$ be a profitable deviation. By assumption, all players in $K$ then deviate to their other option. As a consequence, every edge in $E[K]$ is unicolored in $s'$ if and only if it is unicolored in $s$. Hence
  each cycle in $G[K]$ that is unicolored in $s'$ is also unicolored in
  $s.$ It follows from Corollary \ref{cor:no_unicolor_cycle} that
  $\SW(s') > \SW(s)$. This shows that $\SW$ is a generalized ordinal
  c-potential.
  \qed
\end{proof}

The existence of strong equilibria for coordination
games with two colors and symmetric strategy sets follows from
Proposition 2.2 in \cite{KBW97a}. Corollary~\ref{cor:2-eq} shows that
a stronger result holds, namely that these games have the c-FIP. This
implies that arbitrary coalitional improvement paths always converge
to a strong equilibrium.

We next derive an existence result of $k$-equilibria in graphs in which every pair of cycles is edge-disjoint. We call an edge $e$ of a graph \emph{private} if it belongs to a cycle and is node-disjoint from all other cycles.

\begin{lemma} \label{lem:disj_cycles}
Let $G$ be a graph in which every pair of cycles is edge-disjoint.
Then there exists a private edge.
\end{lemma}
\begin{proof}
Given a cycle $C$, we call a node $v \in V(C)$ an \emph{anchor point} of $C$ if $v$ can be reached from a node $v' \in V(C')$ of another cycle $C' \neq C$ without traversing an edge in $E(C)$. First we show that there always exists a cycle with at most one anchor point. Assume that the claim does not hold. Then every cycle $C$ of $G$ contains at least two distinct anchor points. Fix an arbitrary cycle $C$ of $G$ and let $v^1_C$ and $v^2_C$ be two anchor points of $C$. Start from $v^1_C$ and traverse the edges of $C$ to reach $v^2_C$. Then follow a shortest path $P$ that connects $v^2_C$ to a node $v^1_{C'} \in V(C')$ of another cycle $C' \neq C$; $P$ must exist because $v^2_C$ is an anchor point. 

Note that $C$ and $C'$ share at most one node because all cycles are edge-disjoint; in particular, $P$ might have length zero and consist of a single node only. Because we choose a shortest path connecting $C$ and $C'$, $v^1_{C'}$ must be an anchor point of $C'$. By assumption, $C'$ has another anchor point $v^2_{C'}$. Repeat the above procedure with cycle $C'$ and anchor points $v^1_{C'}$ and $v^2_{C'}$. Continuing this way, we construct a path that traverses cycles of $G$. Eventually, this path must return to a previously visited cycle. So this path contains a cycle and this cycle shares at least one edge with one of the visited cycles. This contradicts the assumption that all cycles of $G$ are pairwise edge-disjoint.

Now, let $C$ be a cycle with at most one anchor point $v$ (if no such node exists, any edge in $E(C)$ is private). Then any edge $e \in E(C)$ such that $v$ is not an endpoint of $e$ is private.
\qed
\end{proof}

\begin{theorem}\label{thm:edge-disjoint}
Consider a coordination game on a graph $G$ in which every pair of cycles is edge-disjoint. Let $k$ be the minimum length of a cycle in $G$. Then every sequence of profitable deviations of coalitions of size at most $k$ is finite. In particular, the game has a 3-equilibrium. 
\end{theorem}

\begin{proof}
We proceed by induction on the number $z$ of cycles. If $z = 1$, then the claim follows by Theorem~\ref{thm:pseudoforest}. 

Now, let $z > 1$. Let $s \betredge{K} s'$ be a profitable deviation such that $|K| \leq k$. From \eqref{eqn:deltasw} we infer that  
$$
\Delta \SW = 2 (\Delta \SW_K - |E^+_{s'} \cap E[K]| + |E^+_s \cap E[K]|)
\ge 2 (\Delta \SW_K - |E[K]|)
$$
because $|E^+_{s'} \cap E[K]| \le |E[K]|$.
Because $k$ is the minimum length of a cycle in $G$ and $|K| \leq k$, we have $|E[K]| \leq |K|.$ So $\Delta \SW \geq 0$.

Consider a sequence of profitable deviations $s_1 \betredge{K_1} s_2 \betredge{K_2} s_3 \ldots$ We show that it is finite. Because the social welfare cannot decrease and is upper bounded there is an index $l \ge 1$ such that for all $i\geq l$, $\SW(s_i) = \SW(s_l).$ We can assume without loss of generality that $l=1$. By Corollary \ref{cor:no_unicolor_cycle}, for each $i \geq 1$ there is a cycle $C_i$ in $G[K_i]$ such $C_i$ is completely non-unicolored in $s_i$ and unicolored in $s_{i+1}.$ Note that $k \leq |V(C_i)| \leq |K_i| \leq k$ and hence $K_i = V(C_i).$

By Lemma \ref{lem:disj_cycles}, there is a cycle $C$ with a private edge $e =\{u_1, u_2\} \in E(C)$. We claim that $C = C_i$ for at most one $i.$ Assume otherwise and let $i_{1}, i_{2}$ such that $C = C_{i_1} = C_{i_2}$ and $C \neq C_i$ for $i_1 < i < i_2$.  Because $C = C_{i_1}$, $e$ is unicolored in $s_{i_1+1}$. We know that $C$ is the only cycle containing $u_j$ for $j=1,2$ by choice of $e$. So $u_j \notin K_i$ for $i_1 < i < i_2$ and hence $e$ is still unicolored in $i_2$. But $C$ switches from completely non-unicolored in $s_{i_2}$ to unicolored in $s_{i_2 + 1}$, a contradiction.

Since $C$ is the only cycle containing $u_j$ for $j = 1, 2$, it follows that each $u_j$ can appear at most once in a deviating coalition. So there is an index $l$ such that $u_1, u_2 \notin K_i$ for all $i>l$. Hence if we remove $e$ and call the new graph $G'$, then for all $i>l$, $s_{i} \betredge{K_i} s_{i+1}$ is a profitable deviation in $G'$. Because $G'$ has one cycle less than $G$, we can apply the induction hypothesis and conclude that the considered sequence of profitable deviations is finite.
\qed
\end{proof}

\subsection{Uniform coordination games}
\label{sec:uniform}

Next, we establish the c-FIP property for some additional classes of coordination games. We call a coordination game on a graph $G$ \emph{uniform} if for every joint strategy $s$ and for every edge $\{i, j\} \in E$ it holds that if $s_i = s_j$ then $p_i(s) = p_j(s)$. 

\begin{theorem}\label{thm:consistent}
Every uniform coordination game has the c-FIP.
\end{theorem}
\begin{proof}
Given a sequence $\theta \in \mathbb{R}^n$ of reals we denote by $\theta^*$ its reordering from the largest to the smallest element. Associate with each joint strategy $s$ the sequence $(p_1(s), \dots, p_n(s))^{*}$ that we abbreviate to $p^{*}(s)$.  We now claim that $p^{*}: S \to \mathbb{R}^n$ is a generalized ordinal c-potential when we take for the partial ordering $\succ$ on $p^{*}(S)$ the lexicographic ordering on the sequences of reals.

Suppose that some coalition $K$ profitably deviates from the joint strategy $s$ to $s' = (s'_K, s_{-K})$.  We claim that then $p^{*}(s') \succ p^{*}(s)$.

Assume this does not hold. Rename the players such that $p^*(s') = (p_1(s'), \dots, p_n(s'))$. Let $i$ be the smallest value for which $p_{i}(s') < p_{i}(s)$. By assumption such an $i$ exists. By the choice of $i$ for all $j < i$ we have $p_{j}(s') \geq p_{j}(s)$ and also $p_{j}(s') \geq p_{i}(s')$.

Now, $p_{i}(s') < p_{i}(s)$ implies that $i \not\in K$ and hence we can write $s' = (s_i, s'_{-i})$. By the definition of the payoff functions, it follows that there exists some neighbor $j$ of $i$ with $s_j = s_i$ and $s'_j \neq s_j$. Thus, $j \in K$. By the uniformity property, $p_i(s) = p_j(s)$. So $p_j(s') > p_j(s) = p_i(s)$. Consequently, by the choice of $i$, we have $p^{*}(s') \succ p^{*}(s)$, which is a contradiction.
\qed
\end{proof}

We can capture by Theorem~\ref{thm:consistent} the following class of coordination games: We say that $G$ is \emph{color complete} (with respect to $A$) if for every $x \in M$ each component of $G[V_x]$ is complete. (Recall that $V_x = \{i \in V \mid x \in A_i\}$.)

\begin{corollary}
\label{cor:color_complete}
Every coordination game on a color complete graph has the c-FIP. In particular, every coordination game on a complete graph has the c-FIP.
\end{corollary} 

The existence of strong equilibria for color complete graphs also follows from a result by Rozenfeld and Tennenholtz \cite{RT06} and the following lemma.

\begin{lemma} 
\label{lem:congestion_iso}
Coordination games on color complete graphs are a special case of monotone increasing congestion games in which all strategies are singletons.
\end{lemma}
\begin{proof}
We can assume without loss of generality that for each color $x$, $G[V_x]$ is connected (otherwise, we replace $x$ with a new respective color for each component of $G[V_x]$). Then we can identify $x$ with a singleton resource, along with the payoff function $v_x: \mathbb N \to \mathbb R$ such that $v_x(k) = k-1$. Now, if a player $i$ chooses $s_i = x$ then
\[
p_i(s) = |\{j \in N_i \mid s_j = x\}| = |\{j \in V \mid s_j = x\}| - 1= v_x(|\{j \in V \mid s_j = x\}|),
\]
so the payoff in the coordination game coincides with the payoff in the associated congestion game. \qed
\end{proof}

Rozenfeld and Tennenholtz \cite{RT06} show that monotone increasing congestion games in which all strategies are singletons admit strong equilibria. Note, however, that our result above is stronger because we show that these games have the c-FIP.

\subsection{Non-existence of $3$-equilibria and existence thresholds}
\label{sec:non-existence-3eq}


We next prove that $3$-equilibria do not exist in general. Recall that $2$-equilibria always exist by Corollary~\ref{cor:2-eq}. 

\begin{figure}[t]
\centering
\begin{tikzpicture}[line join=bevel,z=-5.5,scale=3]
\coordinate (A1) at (0,0,-1); 
\coordinate (A2) at (-1,0,0);
\coordinate (A3) at (0,0,1);
\coordinate (A4) at (1,0,0);
\coordinate (B1) at (0,1,0);
\coordinate (C1) at (0,-1,0);
\coordinate (D1) at (-1.35,0,0);
\coordinate (D2) at (1.35,0,0);
\coordinate (D3) at (0,0,-2);
\coordinate (D4) at (0,0,2);

\draw [very thick,fill opacity=0.5,fill=green] (A2) -- (A1) -- (C1) -- cycle;
\draw [very thick,fill opacity=0.5,fill=yellow] (A2) -- (A3) -- (B1) -- cycle;
\draw [very thick,fill opacity=0.5,fill=red] (A3) -- (C1) -- (A4) -- cycle;
\draw [very thick,fill opacity=0.5,fill=blue] (A4) -- (B1) -- (A1) -- cycle;
\node at (A1) [circle,scale=0.7,fill=black] {};
\node at (A2) [circle,scale=0.7,fill=black] {};
\node at (A3) [circle,scale=0.7,fill=black] {};
\node at (A4) [circle,scale=0.7,fill=black] {};
\node at (B1) [circle,scale=0.7,fill=black] {};
\node at (C1) [circle,scale=0.7,fill=black] {};
\node at (D1) [circle,scale=0.7,fill=black] {};
\node at (D2) [circle,scale=0.7,fill=black] {};
\node at (D3) [circle,scale=0.7,fill=black] {};
\node at (D4) [circle,scale=0.7,fill=black] {};
\draw [very thick] (D1) -- (A2);
\draw [very thick] (D2) -- (A4);
\draw [very thick] (D3) -- (A1);
\draw [very thick] (D4) -- (A3);

\node [above left] at (A1) {\scriptsize $\mathbf{6},\{3,4\}$};
\node [below right] at (A4) {\scriptsize $\mathbf{5},\{2,3\}$};
\node [above right] at (A3) {\scriptsize $\mathbf{4},\{1,2\}$};
\node [below left] at (A2) {\scriptsize $\mathbf{3},\{1,4\}$};
\node [above] at (B1) {\scriptsize $\mathbf{1},\{1,3\}$};
\node [below] at (C1) {\scriptsize $\mathbf{2},\{2,4\}$};
\node [left] at (D1) {\scriptsize $\mathbf{7},\{1\}$};
\node [right] at (D2) {\scriptsize $\mathbf{9},\{3\}$};
\node [above] at (D3) {\scriptsize $\mathbf{10},\{4\}$};
\node [below] at (D4) {\scriptsize $\mathbf{8},\{2\}$};

\end{tikzpicture}
\caption{Three-dimensional illustration of the coordination game used to show that $3$-equilibria do not exist (Theorem~\ref{thm:no3eq}). The colored facets indicate the triangles that can be unicolored. The identity of the players is displayed in boldface. The strategy sets of the players are stated between curly braces.}
\label{fig:octahedron}
\end{figure}
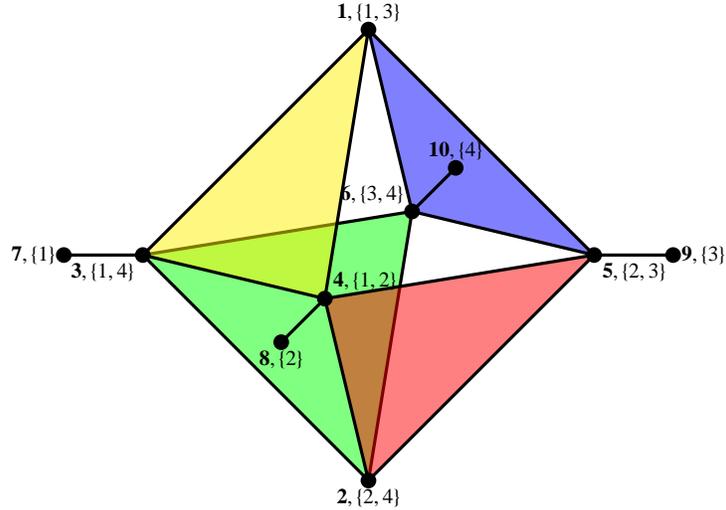

\begin{theorem}\label{thm:no3eq}
There exists a coordination game that does not have a $3$-equilibrium.
\end{theorem}
\begin{proof}
We define a coordination game $\mathcal{G}(G, A)$ as indicated in Figure~\ref{fig:octahedron}: 
There are $n = 10$ players and $4$ colors. The strategy sets are as follows: $A_1 = \{1,3\}$, $A_2 = \{2,4\}$, $A_3 = \{1,4\}$, $A_4 = \{1,2\}$, $A_5 = \{2,3\}$, $A_6 = \{3,4\}$, $A_7 = \{1\}$, $A_8 = \{2\}$, $A_9 = \{3\}$, $A_{10} = \{4\}$. There are $16$ edges, defined as follows: Players $1$ and $2$ are both connected to players $3$,$4$,$5$, and $6$, accounting for $8$ of the edges. There is aditionally a cycle $(3,4,5,6,3)$, accounting for four more edges. Lastly, players $7,8,9$, and $10$ all have a single edge attached to them and are connected to players $3,4,5$, and $6$, respectively. 

As can be seen from Figure~\ref{fig:octahedron}, the graph on which the game is played is essentially the skeleton of an octahedron: $12$ of the edges and $6$ of the nodes of the graph belong to this skeleton, and the four remaining edges are connected to four remaining nodes that are \emph{dummy} players (i.e., they have only one strategy that they can play).

Observe that there are eight triangles in the graph, which correspond to the eight facets of the octahedron. The strategy sets are defined such that only four out of the eight triangles of the octahedron can be unicolored. Also, this game is constructed such that if one triangle is unicolored, then the other three triangles are necessarily not unicolored.

We prove the theorem by showing that for every strategy profile of $\mathcal{G}$, there exists a profitable deviation of a set of at most $3$ players.
To simplify the proof, we make use of the many symmetries in $\mathcal{G}$, which are apparent from Figure \ref{fig:octahedron}.
Let $s$ be an arbitrary strategy profile of $\mathcal{G}$. 
We distinguish two cases:

\begin{itemize}

\item If there is a triangle that is unicolored under $s$, we may assume without loss of generality that this triangle is the one corresponding to players $\{1,3,4\}$ (because of symmetry), i.e., $s_1 = s_3 = s_4 = 1$. Observe that $p_4(s) = 2$. We distinguish two cases:
\begin{itemize}
\item $p_5(s) = 2$. Then $s_5 = s_6 = 3$. If $s_2 = 4$ then player $6$ can deviate profitably to $4$. If $s_2 = 2$ then the coalition $\{2,6\}$ can deviate profitably to $4$. 
\item $p_5(s) \leq 1$. If $s_2 = 2$ and $s_5 = 2$, then player $4$ can deviate profitably to $2$. If $s_2=2$ and $s_5 = 3$ then the coalition $\{4,5\}$ can deviate profitably to $2$. If $s_2 = 4$, then $s_5 = 3$, and coalition $\{2,4,5\}$ can deviate profitably to $2$. 
\end{itemize}

\item If there is no triangle that is unicolored under $s$, we distinguish again two cases. By symmetry we may assume that $s_1 = 1$. 
\begin{itemize}
\item $s_3 = 4$. Then $p_3(s) \leq 1$, so player $3$ can profitably deviate by changing his color to $1$.
\item $s_3 = 1$. Then $s_4 = 2$. If $p_4(s) = 1$ then player $4$ can profitably deviate by changing his color to $1$. Otherwise, $p_4(s) = 2$ and either $(s_2,s_5) = (2,3)$ or $(s_2,s_5) = (4,2)$. 
\begin{itemize}
\item If $(s_2,s_5) = (2,3)$, then if also $p_5(s) = 2$ it holds that $s_6 = 3$ and therefore player $1$ can profitably deviate to $3$. If $p_5(s) = 1$, then player $5$ can profitably deviate to $2$.
\item If $(s_2,s_5) = (4,2)$, then $p_2(s) \leq 1$, so player $2$ can profitably deviate to $2$.
\end{itemize}
\end{itemize} 
\end{itemize}

Note that each profitable deviation given above consists of at most three players. This concludes the proof.
\qed
\end{proof}

The coordination game given in Figure~\ref{fig:octahedron} is an example of a game that does not have a $3$-equilibrium but admits a $2$-equilibrium. 
We define the \emph{transition value} of a coordination game as the value of $k$ for which a $k$-equilibrium exists but a $(k+1)$-equilibrium does not. Clearly, the instance in Figure~\ref{fig:octahedron} has an transition value of $k = 2$. An interesting question is whether one can identify instances of coordination games with a non-trivial transition value $k \ge 3$. 

We next show that the coordination game given in Figure~\ref{fig:non-existence} is an instance with transition value $k = 4$.

\begin{theorem}\label{thm:non-existence}
There is a coordination game that has a transition value of $4$. 
\end{theorem}
\begin{proof}
Consider the coordination game discussed in Example~\ref{exa:1} (see Figure~\ref{fig:non-existence}). We first argue that it does not admit a $5$-equilibrium. Assume for the sake of a contradiction that $s$ is a strong equilibrium of this game.

Consider players 4 and 5.  Let $i, j \in \{4,5\}$, $i \neq j$ be such that $p_i(s) \leq p_j(s)$.  Note that the neighbors of $i$ and $j$ (excluding $j$ and $i$, respectively) are the same. As a consequence, if $s_i \neq s_j$, then player $i$ can profitably deviate to player $j$'s color, i.e., $s'_i = s_j$. Thus players 4 and 5 have the same color in $s$, say $s_4 = s_5 = b$. (Because of the symmetry of the instance, the case $s_4 = s_5 = c$ follows analogously.)

Assume there exists a player $i \in \{2, 3\}$ with $s_i \neq b$. Then $i$ can profitably deviate by choosing $s'_i = b$. It follows that players 2 and 3 have color $s_2 = s_3 = b$.

Next, consider player 8 and suppose $s_8 = b$. Then his payoff is $p_8(s) = 2$. Further, the payoff of each of the players 1, 6 and 7 is 0 because all their neighbors have color $b$. But then the coalition $K = \{1, 6, 7, 8\}$ can profitably deviate by choosing color $a$. We conclude that $s_8 = a$.

As a consequence, for players 1, 6, and 7 we have $s_1 = s_6 = s_7 = a$ as otherwise any such player could profitably deviate by choosing $a$.

Thus, the only remaining possible configuration for $s$ is the one indicted in Figure~\ref{fig:non-existence} (by the underlined strategies). But this is not a strong equilibrium because the coalition $K = \{1, 4, 5, 6, 7\}$ can profitably deviate by choosing color $c$. This yields a contradiction and proves the non-existence of $5$-equilibria. 

On the other hand, it is easy to see that the strategy profile indicated in Figure~\ref{fig:non-existence} constitutes a $4$-equilibrium. This concludes the proof.
\qed
\end{proof}


In general, we leave open the question for which $k \geq 2$ there exist coordination games with transition value $k$.



\bigskip

The above example can be adapted to show that there are coordination games that do not have the c-FIP but are \emph{c-weakly acyclic}. Recall that a game $\mathcal{G}$ is \emph{c-weakly acyclic} if for every joint strategy there exists a finite c-improvement path that starts at it. Note that a c-weakly acyclic game admits a strong equilibrium.

\begin{corollary} 
There is a coordination game that does not have the c-FIP but is $c$-weakly acyclic. 
\end{corollary}
\begin{proof}
Take the coordination game from Example \ref{exa:1} and modify it by adding to each color set a new, common color $d$. Then the joint strategy $s$ in which each player selects $d$ is a strong equilibrium. Moreover, for each player her payoff in $s$  is strictly higher than in any joint strategy in which she chooses another color. So $s$ can be reached from each joint strategy in just one profitable deviation, by a coalition of the players who all switch to $d$. 
On the other hand, the argument presented in Example \ref{exa:1} shows that this game does not have the c-FIP.  
\qed
\end{proof}

\section{Inefficiency of $k$-equilibria}
\label{sec:anarchy}

We first summarize some results concerning the strong price of stability of coordination games. 

\begin{theorem}
The strong price of stability is 1 in each of the following cases:
\begin{itemize}
\item $G$ is a pseudoforest;
\item $G$ is a color forest;
\item there are only two colors.
\end{itemize}
\end{theorem}
\begin{proof}
If $G$ is a pseudoforest, a maximum of $P$ in the lexicographic ordering defined in the proof of Theorem~\ref{thm:pseudoforest} is a strong equilibrium and a social optimum. In the other two cases, the social welfare function $\SW$ is a generalized ordinal c-potential. So in both cases each social optimum is a strong equilibrium. 
\qed
\end{proof}

We next study the $k$-price of anarchy of our coordination games. It is easy to see that the price of anarchy is infinite. In fact, this holds independently of the graph structure, as the next theorem shows.

\begin{theorem} 
For every graph there exists strategy sets for the players such that the price of anarchy of the resulting coordination game is infinite.
\end{theorem}
\begin{proof}
Let $G = (V, E)$ be an arbitrary graph. We assign to each node $i \in V$ a color set $A_i = \{x_i, c\}$, where $x_i$ is a private color, i.e., $x_i \neq x_j$ for every $j \neq i$, and $c$ is a common color. The joint strategy $s$ in which every player chooses her private color constitutes a Nash equilibrium with $\SW(s) = 0$. On the other hand, the joint strategy $s'$ in which every player chooses the common color $c$ is a social optimum with $\SW(s') = 2|E|$. 
\qed
\end{proof}

We now determine the $k$-price of anarchy and the strong price of anarchy. We define for every $j \in N$ and $K \subseteq N$ and joint strategy $s$,
\[
N_j^K(s) = \{\{i,j\} \in E \mid i \in K,\, s_i = s_j\}.
\]
Intuitively, $|N_j^K(s)|$ is the payoff $j$ derives from players in $K$ under $s$.

\begin{theorem} \label{thm:kpoa}
The $k$-price of anarchy of coordination games is between $2\frac{n-1}{k-1}-1$ and $2\frac{n-1}{k-1}$ for every $k \in \{2, \dots, n\}$. Furthermore, the strong price of anarchy is exactly $2$. 
\end{theorem}
\begin{proof}
We first prove the upper bound. 
By the definition of the payoff function for all joint strategies $s$
and $\sigma$, we have
$
|N_j^K(\sigma)| \leq p_j(\sigma_K, s_{-K})
$.

Suppose that the considered game has a $k$-equilibrium, say $s$,
and let $\sigma$ be a social optimum.  By the definition of a
$k$-equilibrium, for all coalitions $K$ of size at most $k$ there
exists some $j \in K$ such that $p_j(\sigma_{K}, s_{-K}) \leq p_j(s)$ and
hence by the above $|N_j^K(\sigma)| \leq p_j(s).$

Fix a coalition $K = \{v_1, \ldots, v_k\}$ of size $k$. We know that there is some $j \in K$ such that $|N_j^{K}(\sigma)| \leq p_j(s)$.
Rename the nodes so that $j = v_k$. Further, there is a node $j$ such
that 
$$
|N_j^{\{v_1, \ldots, v_{k-1}\}}(\sigma)| \leq p_j(s).
$$
Again we
rename the nodes so that $j = v_{k-1}$. Continuing this way
we obtain that for all $i \in \{1, \ldots, k\}$ it holds that $|N_{v_i}^{\{v_1, \ldots, v_i\}}(\sigma)| \leq p_{v_i}(s).$ Hence 
\begin{align*}
  p_{v_i}(\sigma) &= |N_{v_i}^{\{v_1, \ldots, v_i\}}(\sigma)| + |N_{v_i}^{V \setminus \{v_1, \ldots, v_i\}}(\sigma)|
  \leq p_{v_i}(s) + |N_{v_i}^{V \setminus \{v_1, \ldots, v_i\}}(\sigma)|\\
  & = p_{v_i}(s) + |N_{v_i}^{K \setminus \{v_1, \ldots, v_i\}}(\sigma)|
  + |N_{v_i}^{V \setminus K}(\sigma)|.
\end{align*}
Summing over all players in $K$ we obtain
\begin{align} \label{eqn:sw}
\SW_K(\sigma)  \leq \SW_K(s) + \sum_{i = 1}^k \big( |N_{v_i}^{K \setminus \{v_1, \ldots, v_i\}}(\sigma)| + |N_{v_i}^{V \setminus K}(\sigma)|\big).
\end{align}
But
\begin{align*} 
\sum_{i = 1}^k |N_{v_i}^{K \setminus \{v_1, \ldots, v_i\}}(\sigma)|
&= \sum_{i = 1}^k |\{j > i:\, \{v_i, v_j\} \in E^+_{\sigma}\}| = |E^+_\sigma \cap E[K]|
\end{align*}
and
$
\sum_{i = 1}^k |N_{v_i}^{V \setminus K}(\sigma)| = |E^+_\sigma \cap \delta(K)|.
$
Hence rewriting \eqref{eqn:sw} yields
\[
\SW_K(\sigma) \leq \SW_K(s) + |E_\sigma^+ \cap E[K]| + |E_\sigma^+ \cap \delta(K)|.
\]
It also holds that $\SW_K(\sigma) = 2|E^+_\sigma \cap E[K]| + |E^+_\sigma \cap \delta(K)|$. So we get
\[
\SW_K(\sigma) \leq \SW_K(s) + \frac{1}{2} \SW_K(\sigma) + \frac{1}{2} |E^+_\sigma \cap \delta(K)|,
\]
which implies that 
\begin{align} \label{eqn:sw2}
\SW_K (\sigma) \leq 2 \SW_K(s) + |E^+_\sigma \cap \delta(K)|.
\end{align}

Now we sum over all coalitions $K$ of size $k$. Each player $i$
appears in $n-1 \choose k-1$ of such sets because it is possible to
choose $k-1$ out of $n-1$ remaining players to form a set $K$ of
size $k$ that contains $i$. Hence,
\begin{align*}
\sum_{K: |K| = k} \SW_K(\sigma) &
= \sum_{i = 1}^{n} \sum_{K: \, K \ni i} p_i(\sigma) 
= \sum_{i = 1}^{n} {n-1 \choose k-1} p_i(\sigma) = {n-1 \choose k-1} \SW(\sigma).
\end{align*}
We obtain an analogous expression for the joint strategy $s$.

Furthermore, for each edge $e = \{u,v\} \in E^+_\sigma$, we can choose
$2 {n-2 \choose k-1}$ sets $K$ of size $k$ such that $e \in
\delta(K)$. Indeed, assuming that $u \in K$ and $v \notin K$, we can
choose $k-1$ out of $n-2$ remaining players to complete $K$ and hence
there exist ${n-2 \choose k-1}$ of those sets. Reversing the roles of
$u$ and $v$ and summing up yields $2 {n-2 \choose k-1}$.  Hence
\[
\sum_{K: |K| = k} |E^+_\sigma \cap \delta(K)| = 2 {n-2 \choose k-1} |E^+_\sigma| = {n-2 \choose k-1} \SW(\sigma).
\]
By summing over all coalitions $K$ of size $k$, equation \eqref{eqn:sw2} yields
\[
{n-1 \choose k-1} \SW(\sigma) \leq 2 {n-1 \choose k-1} \SW(s) + {n-2 \choose k-1} \SW(\sigma).
\]
It follows that the $k$-price of anarchy is at most 
$$
\frac{2{n-1 \choose k-1}}{{n-1 \choose k-1} - {n-2 \choose k-1}} = 2 \frac{n-1}{k-1}.
$$
This concludes the proof of the upper bound. 

The claimed lower bounds follow from Examples~\ref{ex:lb1} and \ref{ex:lb2} given below.
\qed
\end{proof}

The following example establishes a lower bound on the $k$-price of anarchy. 

\begin{example}\label{ex:lb1}
Fix $n$ and $k \in \{2, \dots, n\}$. Let $V(G)$ consist of two sets $V_1$ and $V_2$ of size $k$ and $n-k$, respectively, and define
\[
E[G] = \{\{u,v\} \mid u \in V_1, v \in V_1 \cup V_2\}.
\]
Fix three colors $a,b$ and $c$. For $v \in V_1$, let $A(v) = \{a,c\}$. For $v \in V_2$, let $A(v) = \{b,c\}$. Then the color assignment $\sigma$ in which each player chooses the common color $c$ is a social optimum. The social welfare is
\[
\SW(\sigma) = \SW_{V_1}(\sigma) + \SW_{V_2}(\sigma) = k(n-1) + (n-k)k.
\]

Next we show that the color assignment $s$ in which every node in $V_1$ chooses $a$ and every node in $V_2$ chooses $b$ is a $k$-equilibrium. Assume that there is a profitable deviation $s \betredge{K} s'$ such that $|K| \leq k$. Then all nodes in $K$ switch to $c$ and also all nodes that choose $c$ in $s'$ are in $K$. Hence for all $v \in K$, $p_v(s') = |N_v \cap K|.$ So there is a node $v \in V_1 \cap K$ because otherwise the payoff of all nodes in $K$ would remain 0. But then $p_v(s') = |N_v \cap K| \leq k = p_v(s)$, which yields a contradiction. 

Note that $\SW(s) = k(k-1).$
It follows that the $k$-price of anarchy is at least
\[
\frac{\SW(\sigma)}{\SW(s)} = \frac{k(n-1) + (n-k) k}{k(k-1)} = \frac{2 (n-1) - (k-1)}{k-1} = 2\frac{n-1}{k-1}-1.
\]
\qed
\end{example}

The following example shows that the upper bound of $2$ on the strong price of anarchy ($k = n$) of Theorem~\ref{thm:kpoa} is tight. 

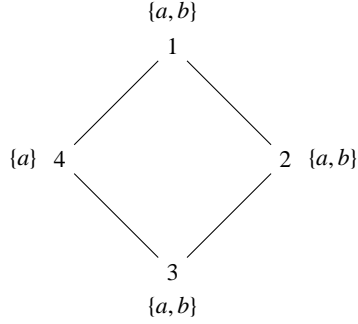
\begin{figure}[t]
\begin{center}
\begin{tikzpicture}[scale=1.5]
    \node [label={\lf $\{ a, b\}$}] (1) at (0,1) {$1$};    
    \node [label=right:{\lf $\{ a, b\}$}] (2) at (1,0) {$2$};    
    \node [label=below:{\lf $\{ a, b\}$}] (3) at (0,-1) {$3$};    
    \node [label=left:{\lf $\{ a \}$}] (4) at (-1,0) {$4$};       

       \draw (1) -- (2);       
       \draw (2) -- (3);              
       \draw (3) -- (4);              
       \draw (1) -- (4);       
\end{tikzpicture}
\end{center}
\caption{A coordination game showing that the strong price of anarchy is at least 2.}
\label{fig:graph2}
\end{figure}

\begin{example}\label{ex:lb2}
Consider the graph and the color assignment depicted in Figure~\ref{fig:graph2}.
Here $(a,a,a,a)$ is a social optimum with the social welfare 8, while
$(b,b,b,a)$ is a strong equilibrium with the lowest social welfare, 4. So the strong price of anarchy is $2$ in this example of $4$ players. By duplicating the graph $l$ times, we can draw the same conclusion for the case of $4l$ players. 
\qed
\end{example}


\section{Complexity}
\label{sec:computation} 

In this section we study complexity issues concerning $k$-equilibria.

\subsection{Verification}

First, we show that in general it is hard to decide whether a given joint strategy is a $k$-equilibrium. 

Let $k$-\textsc{Equilibrium} denote the problem to decide, given a coordination game with a joint strategy $s$ and $k \in \{1, \dots, n\}$, whether $s$ is a $k$-equilibrium. 

\begin{theorem}
$k$-\textsc{Equilibrium} is co-NP-complete.
\end{theorem}
\begin{proof}
It is easy to verify that  $k$-\textsc{Equilibrium} is in co-NP: a certificate of a NO-instance is a profitable deviation of a coalition of size at most $k$.

We show the hardness by reduction of the complement of \textsc{Clique}, which is a co-NP-complete problem. Let $(G,k)$ be an instance thereof. We construct an instance of  $k$-\textsc{Equilibrium} as follows. For $v \in V$ let $A_v = \{x_v, y\}$, where color $y$ and all colors $x_v, v\in V$ are distinct. Furthermore, for every node $v \in V$ we add $k-2$ nodes $u_v^1, \ldots, u_v^{k-2}$ and edges $\{v,u_v^i\}$ for $i = 1, \ldots, k-2$. These additional nodes can only choose the color $x_v.$ Let $s$ be the joint strategy in which every node $v \in V$ chooses $x_v$. We claim that this is a $k$-equilibrium if and only if $G$ has no clique of size $k.$

Suppose $G$ has a clique $K$ of size $k$. Then jointly deviating to $y$ yields to each node in $K$ a payoff of $k-1$, whereas every node has a payoff of $k-2$ in $s$. So this is a profitable deviation. For the other direction, suppose that there is a profitable deviation $s \betredge{K} s'$ by a coalition $K$ of size at most $k$. Then every node in $K$ deviates to $y$ and hence belongs to $V.$ Since every node in $K$ has a payoff of $k-2$ in $s$, $p_v(s') \geq k-1$ for all $v \in K.$ So $v$ is connected to at least $k-1$ nodes in $K$. This implies that $K$ is a clique of size $k.$
\qed
\end{proof}

We next show that for color forests the decision problem is in \cP. First we show that we can focus on certain profitable deviations which we call \emph{simple}: Fix a joint strategy $s$ and a coalition $K$. We call $K$ \emph{connected} if $G[K]$ is connected. A deviation $s \betredge{K} s'$ is \emph{simple} if $K$ is connected and $s' = (x_K, s_{-K})$ for some color $x.$

\begin{lemma}\label{lem:simple}
Let $s$ be a joint strategy in a coordination game. If there is a profitable deviation by a coalition of size at most $k$, then there is also a simple profitable deviation by a coalition of size at most $k$.
\end{lemma}
\begin{proof}
Let $s \betredge{K} s'$ be a profitable deviation with $|K| \leq k$. Pick an arbitrary $v \in K$ and let $x = s'_v$. Let $L$ consist of those nodes $u \in K$ for which $s'_u = x$ and $u$ is reachable in $G[K]$ from $v.$ Let $s'' = (x_L, s_{-L})$. 
 Then the deviation to $s''$ is simple. For all nodes $u\in L$, we have $N_u^K(s') = N_u^L(s') = N_u^L(s'')$ by the definition of $L$. Furthermore,
$
N_u^{V \setminus K}(s') \subseteq N_u^{V \setminus L}(s') \subseteq N_u^{V \setminus L}(s'').
$
Hence 
\[
p_u(s') = |N_u^K(s')| + |N_u^{V \setminus K}(s')| \leq |N_u^L(s'')| + |N_u^{V \setminus L}(s'')| = p_u(s''),\]
which implies that the deviation to $s''$ is profitable for $u.$ \qed
\end{proof}

\begin{theorem}\label{thm:computation_color_forest}
Consider a coordination game on a color forest. Then there exists a polynomial-time algorithm that decides whether a given joint strategy is a $k$-equilibrium and, if this is not the case, outputs a profitable deviation of a coalition of size at most $k.$ 
\end{theorem}
\begin{proof} 
For a statement $P$ we write below $\tlb P \trb$ to denote the variable that is 1 if $P$ is true and $0$ otherwise. For a function $f: V \to \mathbb R$ and $U \subseteq V$, let $f(U) = \sum_{v \in U} f(v).$ For a function $F: V \to 2^V$ let $F(U) = \bigcup_{v \in U} F(v).$

Let $s$ be a joint strategy. By Lemma \ref{lem:simple} it is
sufficient to check for the existence of \emph{simple} profitable
deviations by coalitions of size at most $k$. Thus, we let $x \in M$ and we
search for simple profitable deviations in which the coalition deviates to
$x$. Because a coalition in a simple deviation is connected, we can
check each connected component of $G[V_x]$ separately. Assume
without loss of generality that $G[V_x]$ itself is connected, i.e., is
a tree. Pick an arbitrary root $r$ of $G[V_x]$ and define for each node the \emph{children}, \emph{parent}, and \emph{rooted subtree} in the usual way (with respect to $r$). For each node $v
\in V_x$ let $\m C_v \su V_x$ denote the set of children of $v$ and
let $\m P_v \in V_x$ denote the parent of $v$ (if $v \neq
r$). Finally, let $\m T_v$ denote the subtree of $G[V_x]$ rooted at
$v$. For each node $v$, we define $\m U(v)$, $D(v) $, $\m U^p(v)$, and $D^p(v)$ as follows.

\begin{itemize}
\item $\m U(v)$ is a connected coalition $K \subseteq \m T_v$  of minimum size such that $v \in K$ and the deviation to $(x_{K},s_{-K})$ is profitable for all nodes in $K$ (if such a coalition exists). We denote the properties it has to satisfy by ($*$). 
\item $D(v) = |\m U(v)|$ if $\m U(v)$ exists and $\infty$ otherwise.
\item $\m U^p(v)$ is a connected coalition $L \subseteq \m T_v$ of minimum size such that 
$v \in L$ and the deviation to $(x_{L'}, s_{-L'})$ is profitable for all nodes in $L$, where $L' := L \cup \{\m  P_v\}$ (if such a coalition exists).  We denote the properties it has to satisfy by ($**$).  (Note that the deviation is not required to be profitable for $\m P_v$ even though $\m P_v \in L'.$)
\item $D^p(v) = |\m U^p(v)|$ if $\m U^p(v)$ exists and $\infty$ otherwise.
\end{itemize}

We can compute $D$, $D^p$, $\m U$ and $\m U^p$ using a dynamic program as follows. Let $v \in V_x$ and suppose we found these objects for all children of $v$. Let $U \subseteq \m C_v$ minimize $D^p(U)$ among all sets $U' \subseteq \m C_v$ that satisfy $D^p(U') < \infty$ and
\begin{eqnarray}\label{eqn:1}
|U'| +  \tlb s_{\m P_v} = x \trb > p_v(s)
\end{eqnarray}
if such a set exists. In this case set $\m U(v) = \{v\} \cup \m U^p(U)$ and $D(v) = |\m U(v)|$. Otherwise, we set it to $\infty$.

We first prove that $K := \m U(v)$  satisfies  ($*$) if $\m U(v)$ exists. The set $K$ is connected because $v \in K$ and all sets $\m U^p(u)$ are connected, for $u \in \m C_v$. It is profitable for $v$ to deviate to $x$ because of \eqref{eqn:1}.  The deviation is profitable for nodes $u \in K \setminus \{v\}$ because $\m U^p(u) \su K$, $\m P(u) \in K$ (by the connectivity of $K$) and $\m U^p(u)$ satisfies ($**$). Furthermore, $K$ is of minimal size amongst all coalitions that satisfy ($*$). Indeed, if $K'$ is another such coalition, then $U' := K \cap \m C_v$ satisfies \eqref{eqn:1} because it is profitable to deviate to $x$ for $v.$ It is profitable for $u \in U'$ to deviate to $x$ and hence $|K \cap \m T(u)| \geq D^p(u)$, which implies $|K'| \geq 1+ D^p(U').$ Therefore $|K'| \geq 1 + D^p(U)$ by the minimality of $U$. But $|K| = |\{v\} \cup \m U^p(U)| = 1 + D^p(U)$, which shows that $|K'| \geq |K|.$

Similarly, for $v \in V_x \setminus \{r\}$, let $W \subseteq \m C_v$ minimize $D^p(W)$ among all sets $W' \subseteq \m C_v$ that satisfy $D^p(W') < \infty$ and
\begin{equation}\label{eqn:2}
|W'| + 1 > p_v(s)
\end{equation}
if such a set exists. In this case set $\m U^p(v) = \{v\} \cup \m U^p(W)$ and $D^p(v) = |\m U^p(v)|$. Otherwise, we set $D^p(v) = \infty.$ Similar arguments as before show that if $\m U^p(v)$ exists then it indeed satisfies ($**$).

Note that we can compute $\m U(v)$ in polynomial time by sorting the nodes $u \in \m C_v$ in increasing order of $D^p(u)$ and then successively adding nodes to $\m U(v)$ until \eqref{eqn:1} is satisfied. Similarly, we can compute $\m U^p(v)$ efficiently. This shows that the algorithm runs in polynomial time.

Now, let $s \betredge{K} s'$ be a simple profitable deviation to $x$ such that $|K| \leq k$. Let $v$ be the root of $G[K]$ according to our previously fixed ordering. By the properties of the function $D$, we know that $D(v) \leq |K| \leq k$. Conversely, if $D(v) \leq k$ for some node $v$, then $\m U(v)$ of size $D(v) \leq k$ is the coalition we are looking for. 
\qed
\end{proof}

\subsection{Computing strong equilibria}

Next we focus on the problem of actually computing a strong equilibrium. As we show below, this is possible for certain graph classes.

\begin{corollary}
\label{cor:computation_color_forest}
Consider a coordination game on a color forest. Then a strong equilibrium can be computed in polynomial time. 
\end{corollary}
\begin{proof}
We begin with an arbitrary initial joint strategy $s$. Putting $k = n$, by Theorem \ref{thm:computation_color_forest} there is an algorithm that decides whether $s$ is a strong equilibrium and, if this is not the case, outputs a profitable deviation $s \betredge{K} s'.$ In the first case, we output $s$; in the second case, we repeat the procedure with $s'$. We know that $\SW(s)$ is a natural number and $\SW(s') > \SW(s)$ by Corollary \ref{thm:color_forest}, so at most $\max_{s \in S} \SW(s) \leq 2|E|$ steps are necessary to reach a strong equilibrium.
\qed
\end{proof}

\begin{theorem}
Consider a coordination game on a color complete graph. Then a strong equilibrium can be computed in polynomial time.
\end{theorem}
\begin{proof}
This follows from Lemma \ref{lem:congestion_iso} and the corresponding result for monotone increasing congestion games in which all strategies are singletons, established in \cite{RT06}.
\qed
\end{proof}

\begin{theorem} \label{thm:computation_pseudoforest} 
Consider a coordination game on a pseudoforest. Then a strong equilibrium can be computed in  polynomial time. 
\end{theorem}
\begin{proof}
We first show that for a tree, a strong equilibrium can be computed efficiently via dynamic programming. By Corollary~\ref{thm:color_forest} it suffices to compute a social optimum. Let $T$ be a tree and root it at an arbitrary node $r \in V(T)$. Given a node $i \in V(T)$, let $\m T_i$ denote the subtree of $T$ that is rooted at $i$ and let $\m  C_i$ be the set of  children of $i$. Given a color $s_i \in S_i$, define $d_i(s_i)$ as the maximum social welfare achievable by the nodes in $\m T_i$ if node $i$ chooses color $s_i$. Note that for each leaf $i \in V(T)$ of $T$ we have $d_i(s_i) = 0$ for all $s_i \in S_i$. 
Consider a node $i \in V(T)$ that is not a leaf and assume we computed all values $d_j(s_j)$ for every $j \in \m C_i$ and $s_j \in S_j$. Define $\tlb s_j = s_i \trb$ to be 1 if $s_j = s_i$ and 0 otherwise.
We can then compute $d_i(s_i)$ for every $s_i \in S_i$ as follows: 
$$
d_i(s_i) = \sum_{\text{$j \in \m  C_i$}} \max_{s_j \in S_j} (d_j(s_j) + 2  \tlb s_j = s_i \trb).
$$
The intuition here is that we account for every child $j \in \m  C_i$ of $i$ for the maximum social welfare achievable in $\m T_j$ plus an additional contribution of $2$ if $i$ and $j$ choose the same color.

Computing $d_i(s_i)$ for all $s_i \in S_i$ takes time at most $O(m^2|\m  C_i|)$, where $m$ is the number of colors. Thus, it takes time $O(m^2 |V(T)|)$ to compute all values $d_r(s_r)$ for $s_r \in S_r$ of the root node $r$.
The optimal social welfare of the tree $T$ is then $\SW(T) = \max_{s_r \in S_r} d_r(s_r)$.
The corresponding optimal joint strategy $s^*_T$ can be determined using some standard bookkeeping.

Next suppose that $T$ is a pseudotree. Let $C = (i_1, \dots, i_k)$ be the unique cycle in $T$. Note that it might no longer be sufficient to simply compute a social optimum for $T$. Instead, the idea is to compute a social optimum $s^*_T$ of $T$ such that, if possible, $C$ is unicolored. 

Note that if such a social optimum does not exist, then there is an edge in $C$ that is not unicolored. Let $\SW(j)$, $j \in \{1, \dots, k\}$, be the maximum social welfare of the tree that one obtains from $T$ by removing edge $\{i_j, i_{j+1}\}$ from $C$ (where we define $i_{k+1} = i_1$). Note that we can efficiently compute $\SW(j)$ by using the dynamic program for trees described above.\footnote{Observe that we do not enforce that the endpoints of the removed edge $\{i_j, i_{j+1}\}$ obtain different colors in the optimal solution. In fact, subsequently it will become clear that we do not have to do so.} Let $\SW_1 = \max_{j = 1, \dots, k} \SW(j)$. Computing $\SW_1$ takes time $O(k \cdot m^2 |V(T)|) = O(nm^2 |V(T)|)$.

Next assume a social optimum exists in which all nodes of $C$ are unicolored. Note that if we remove the edges on $C$ from $T$ then $T$ decomposes into $k$ trees, rooted at $i_1, \dots, i_k$. We can compute $d_{i_j}(\cdot)$ for every root $i_j$ as described above. 
Let $R = \cap_{j = 1}^k S_{i_j}$ be the set of common colors of the nodes in $C$. 
If all nodes in $C$ choose color $c \in R$ then we obtain a social welfare of 
$$
\SW(c) = 2k + \sum_{j = 1}^k d_{i_j}(c).
$$
Let $\SW_2 = \max_{c \in R} \SW(c)$. The time needed to compute $\SW_2$ is at most $O(m^2 |V(T)| + k \cdot m)$.

Clearly, if $\SW_1 > \SW_2$ then there is no social optimum in which all nodes of $C$ have the same color. In this case, we choose an arbitrary social optimum. Otherwise, there exists a social optimum in which all nodes of $C$ have a common color. In this case, we choose such a social optimum. Let the resulting social optimum for pseudotree $T$ be $s^*_{T}$.

By proceeding this way for each pseudotree $T$ of the given pseudoforest $G$, we obtain a joint strategy $s^*$ that maximizes the social welfare and the number of unicolored cycles. By Corollary~\ref{cor:one-cycle}, $s^*$ is a strong equilibrium of $G$. The time needed per pseudotree $T$  is dominated by $O(n m^2 |V(T)|)$. The total time needed to compute $s^*$ is thus at most $O(n^2 m^2)$. 
\qed
\end{proof}

\section{Conclusions}
\label{sec:conc}

We introduced and studied a natural class of games which we termed coordination games on graphs. We provided results on the existence, inefficiency and computation of strong equilibria for these games. 


It would be interesting to prove existence of $k$-equilibria for other graph classes and to investigate the computational complexity of computing them. 
Another open question is to determine the (strong) price of anarchy
when the number of colors is fixed.  Yet another intriguing question
is for which $k \geq 2$ coordination games with transition value $k$
exist. In Section 4.5 we settled this question positively only for $k
= 2$ and $k = 4$.  In the future we also plan to study a natural
extension of our coordination games to hypergraphs.

Another natural question that comes to one's mind is whether \emph{super strong equilibria} exist. Recall that a joint strategy $s$ is a \emph{super strong equilibrium} if for all coalitions $K$ there does not exist a deviation $s' = (s'_K, s_{-K})$ such that $p_i(s') \geq p_i(s)$ for all $i \in K$ and $p_i(s') > p_i(s)$ for some $i \in K$. It is not hard to verify that super strong equilibria are not guaranteed to exist: Consider a path consisting of two edges and assume that the nodes have color sets $\{a\}, \ \{a, b\}$ and $\{b\}$, respectively. Clearly, a super strong equilibrium does not exist for this instance.


A natural generalization of our model are coordination games on weighted graphs. Here each edge $\{i, j\}$ has a non-negative weight $w_{ij}$ specifying how much player $i$ and $j$ profit from choosing the same color. It is easy to see that $\frac{1}{2}\SW$ continues to be an exact potential function for weighted coordination games, guaranteeing the existence of a Nash equilibrium. In fact, as observed in \cite{CD11}, this is an exact potential for coordination games with arbitrary weights. Coordination games on weighted graphs are studied in more detail in \cite{rahn:schaefer:2015}. In particular, the existence results for strong equilibria (Theorems \ref{thm:color_forest} and \ref{thm:pseudoforest}) and 2-equilibria (Corollary \ref{cor:2-eq}) do not hold for these games. We refer the reader to \cite{rahn:schaefer:2015} for further studies of these games.

Another natural variation is to consider coordination games on
  weighted \textit{directed graphs}. Given a directed graph
  $G=(V,E)$, we say that node $j$ is a \textit{neighbour} of node
  $i$ if there is an edge $(j,i)$ in $G$. Each edge $(j,i)$ has a
  non-negative weight $w_{ji}$ specifying how much player $i$ profits
  from choosing the same color as player $j$. The transition from
  undirected to directed graphs changes the status of the games
  substantially. In particular, Nash equilibria need not always exist
  in these games. Moreover, the problem of determining the existence
  of Nash equilibria is NP-complete. We refer the reader to
  \cite{ASW15} and \cite{SW16} for further studies of these games.  

\subsection*{Acknowledgements}

The fact that for the case of a ring the coordination game has the
c-FIP was first observed by Dariusz Leniowski. We thank Jos\'e Correa
for allowing us to use his lower bound in Theorem~\ref{thm:kpoa}. It
improves on our original one by a factor of $2$.  We thank the
anonymous reviewers for their valuable comments.  First author is also
a Visiting Professor at the University of Warsaw.  He was partially
supported by the NCN grant nr 2014/13/B/ST6/01807.

\bibliographystyle{spmpsci}
\bibliography{e,apt,clustering}

\begin{thebibliography}{10}
\providecommand{\url}[1]{{#1}}
\providecommand{\urlprefix}{URL }
\expandafter\ifx\csname urlstyle\endcsname\relax
  \providecommand{\doi}[1]{DOI~\discretionary{}{}{}#1}\else
  \providecommand{\doi}{DOI~\discretionary{}{}{}\begingroup
  \urlstyle{rm}\Url}\fi

\bibitem{AFM09}
Andelman, N., Feldman, M., Mansour, Y.: Strong price of anarchy.
\newblock Games and Economic Behavior \textbf{65}(2), 289--317 (2009)

\bibitem{AM11}
Apt, K.R., Markakis, E.: Diffusion in social networks with competing products.
\newblock In: Proceedings of the 4th International Symposium on Algorithmic
  Game Theory (SAGT), \emph{Lecture Notes in Computer Science}, vol. 6982, pp.
  212--223. Springer (2011)

\bibitem{ARSS14}
Apt, K.R., Rahn, M., Sch{\"a}fer, G., Simon, S.: Coordination games on graphs
  (extended abstract).
\newblock In: Proceedings of the 10th Conference on Web and Internet Economics
  (WINE), \emph{Lecture Notes in Computer Science}, vol. 8877, pp. 441--446.
  Springer (2014)

\bibitem{AS13}
Apt, K.R., Simon, S.: Social network games with obligatory product selection.
\newblock In: Proceedings 8th International Symposium on Games, Automata,
  Logics and Formal Verification (GandALF), vol. 119, pp. 180--193. Electronic
  Proceedings in Theoretical Computer Science (2013)

\bibitem{ASW15}
Apt, K.R., Simon, S., Wojtczak, D.: Coordination games on directed graphs.
\newblock In: Proc. of the 12th Conference on Theoretical Aspects of
  Rationality and Knowledge (TARK 2015) (2015)

\bibitem{Aumann59}
Aumann, R.J.: Acceptable points in general cooperative n-person games.
\newblock In: R.D. Luce, A.W. Tucker (eds.) Contribution to the theory of game
  IV, Annals of Mathematical Study 40, pp. 287--324. University Press (1959)

\bibitem{Aziz}
Aziz, H., Brandt, F.: Existence of stability in hedonic coalition formation
  games.
\newblock In: Proceedings of the 11th International Conference on Autonomous
  Agents and Multiagent Systems (AAMAS), pp. 763--770 (2012)

\bibitem{ABS10}
Aziz, H., Brandt, F., Seedig, H.G.: Optimal partitions in additively separable
  hedonic games.
\newblock In: Proceedings of the 3rd International Workshop on Computational
  Social Choice (COMSOC), pp. 271--282 (2010)

\bibitem{ABS11}
Aziz, H., Brandt, F., Seedig, H.G.: Stable partitions in additively separable
  hedonic games.
\newblock In: Proceedings of the 10th International Conference on Autonomous
  Agents and Multiagent Systems (AAMAS), pp. 183--190 (2011)

\bibitem{Baner}
Banerjee Konishi, S.: Core in a simple coalition formation game.
\newblock Social Choice and Welfare \textbf{18}, 135--153 (2001)

\bibitem{BFFM11}
Bil{\`o}, V., Fanelli, A., Flammini, M., Moscardelli, L.: Graphical congestion
  games.
\newblock Algorithmica \textbf{61}(2), 274--297 (2011)

\bibitem{Bogo}
Bogomolnaia, A., Jackson, M.O.: The stability of hedonic coalition structures.
\newblock Games and Economic Behavior \textbf{38}(2), 201–230 (2002)

\bibitem{CD11}
Cai, Y., Daskalakis, C.: On minmax theorems for multiplayer games.
\newblock In: Proceedings of the 22nd Annual ACM-SIAM Symposium on Discrete
  Algorithms, pp. 217--234 (2011)

\bibitem{CKPS10}
Chatzigiannakis, I., Koninis, C., Panagopoulou, P.N., Spirakis, P.G.:
  Distributed game-theoretic vertex coloring.
\newblock In: Proceedings of 14th International Conference on Principles of
  Distributed Systems (OPODIS), \emph{Lecture Notes in Computer Science}, vol.
  6490. Springer (2010)

\bibitem{EGM12}
Escoffier, B., Gourv{\`e}s, L., Monnot, J.: Strategic coloring of a graph.
\newblock Internet Mathematics \textbf{8}(4), 424--455 (2012)

\bibitem{FF15}
Feldman, M., Friedler, O.: A unified framework for strong price of anarchy in
  clustering games.
\newblock In: Automata, Languages, and Programming, \emph{Lecture Notes in
  Computer Science}, vol. 9135, pp. 601--613. Springer (2015)

\bibitem{FLN15}
Feldman, M., Lewin-Eytan, L., Naor, J.S.: Hedonic clustering games.
\newblock ACM Transactions on Parallel Computing \textbf{2}(1), 4:1--48 (2015)

\bibitem{Gair}
Gairing, M., Savani, R.: Computing stable outcomes in hedonic games.
\newblock In: Proceedings of the 3rd International Symposium on Algorithmic
  Game Theory (SAGT), pp. 174--185 (2010)

\bibitem{GM09}
Gourv{\`{e}}s, L., Monnot, J.: On strong equilibria in the max cut game.
\newblock In: Proc. 5th International Workshop on Internet and Network
  Economics, {WINE}, \emph{Lecture Notes in Computer Science}, vol. 5929, pp.
  608--615. Springer (2009)

\bibitem{GM10}
Gourv{\`{e}}s, L., Monnot, J.: The max \emph{k}-cut game and its strong
  equilibria.
\newblock In: Proceedings of the , 7th Annual Conference on the Theory and
  Applications of Models of Computation {TAMC}, \emph{Lecture Notes in Computer
  Science}, vol. 6108, pp. 234--246. Springer (2010)

\bibitem{Harks13}
Harks, T., Klimm, M., M{\"o}hring, R.: Strong equilibria in games with the
  lexicographical improvement property.
\newblock International Journal of Game Theory \textbf{42}(2), 461--482 (2013)

\bibitem{Hoefer2007}
Hoefer, M.: Cost sharing and clustering under distributed competition (2007).
\newblock Ph.D. Thesis, University of Konstanz, 2007. Available from
  \url{www.mpiinf.mpg.de/\~{}mhoefer/05-07/diss.pdf}

\bibitem{Holzman97}
Holzman, R., Law-Yone, N.: Strong equilibrium in congestion games.
\newblock Games and Economic Behavior \textbf{21}(1--2), 85--101 (1997)

\bibitem{How72}
Howson, J.: Equilibria of polymatrix games.
\newblock Management Science \textbf{18}(5), 312--318 (1972)

\bibitem{JZ12}
Jackson, M., Zenou, Y.: Games on networks (2012).
\newblock Centre for Economic Policy Research Discussion Paper No. 9127. 86
  pages

\bibitem{Jan68}
Janovskaya, E.: Equilibrium points in polymatrix games.
\newblock Litovskii Matematicheskii Sbornik \textbf{8}, 381--384 (1968)

\bibitem{Kon27}
K{\"{o}}nig, D.: \"{U}ber eine {S}chlu\ss weise aus dem {E}ndlichen ins
  {U}nendliche.
\newblock Acta Litt. Ac. Sci. \textbf{3}, 121--130 (1927)

\bibitem{KBW97a}
Konishi, H., {Le Breton}, M., Weber, S.: Pure strategy {Nash} equilibrium in a
  group formation game with positive externalities.
\newblock Games and Economic Behaviour \textbf{21}, 161--182 (1997)

\bibitem{Mil96}
Milchtaich, I.: Congestion games with player-specific payoff functions.
\newblock Games and Economic Behaviour \textbf{13}, 111--124 (1996)

\bibitem{PS08}
Panagopoulou, P.N., Spirakis, P.G.: A game theoretic approach for efficient
  graph coloring.
\newblock In: Proceedings of the 19th International Symposium on Algorithms and
  Computation, (ISAAC), \emph{Lecture Notes in Computer Science}, vol. 5369,
  pp. 183--195. Springer (2008)

\bibitem{rahn:schaefer:2015}
Rahn, M., Sch\"{a}fer, G.: Efficient equilibria in polymatrix coordination
  games.
\newblock In: Proceedings of the 40th International Symposium on Mathematical
  Foundations of Computer Science (MFCS), \emph{Lecture Notes in Computer
  Science}, vol. 9235, pp. 529--541 (2015)

\bibitem{Ros73}
Rosenthal, R.W.: A class of games possessing pure-strategy {Nash} equilibria.
\newblock International Journal of Game Theory \textbf{2}(1), 65--67 (1973)

\bibitem{RT06}
Rozenfeld, O., Tennenholtz, M.: Strong and correlated strong equilibria in
  monotone congestion games.
\newblock In: Proceedings of the 2nd International Workshop on Internet and
  Network Economics (WINE), \emph{Lecture Notes in Computer Science}, vol.
  4286, pp. 74--86. Springer (2006)

\bibitem{SA15}
Simon, S., Apt, K.R.: Social network games.
\newblock Journal of Logic and Computation \textbf{25}(1), 207--242 (2015)

\bibitem{SW16}
Simon, S., Wojtczak, D.: Efficient local search in coordination games on
  graphs.
\newblock In: Proceedings of the Twenty-Fifth International Joint Conference on
  Artificial Intelligence (IJCAI), pp. 482--488. AAAI Press (2016)

\end{thebibliography}

\appendix

\section{c-FIP and generalized ordinal c-potentials}
\label{app:gen-ord-potential}

\begin{theorem}
A finite game has the c-FIP iff a generalized ordinal c-potential for it exists. 
\end{theorem}

\begin{proof}
  $(\Rightarrow)$ We use here the argument given in the proof of \cite{Mil96}
  of the fact that every finite game that has the FIP (finite
  improvement property) has a generalized ordinal potential.
  
  Consider a branching tree of which the root has all joint strategies
  as successors, of which the non-root elements are joint strategies,
  and of which the branches are the c-improvement paths. Because
  the game is finite, this tree is finitely branching.
  
  K\"{o}nig's Lemma of \cite{Kon27} states that any finitely branching
  tree is either finite or it has an infinite path.  So by the
  assumption, the considered tree is finite. Hence the
  number of c-improvement paths is finite.  Given a joint strategy $s$,
  define $P(s)$ to be the number of prefixes of the c-improvement paths
  that terminate in $s$.  Then $P$ is a generalized ordinal c-potential,
where we use the strict linear ordering on the natural numbers.

$(\Leftarrow)$ Immediate, as already noted in  \cite{Holzman97}. 
\qed
\end{proof}

\end{document}